\documentclass{llncs}
\usepackage{graphicx,amssymb,amsmath}

\title{Recognition of Triangulation Duals of Simple Polygons With and Without Holes}

\author{Martin Derka\thanks{The first author was supported by Vanier CGS.} \and Alejandro L\'{o}pez-Ortiz
    \and Daniela Maftuleac}
\institute{David R. Cheriton School of Computer Science, University of Waterloo\\\email{\{mderka,alopez-o,dmaftule\}@uwaterloo.ca}}

\index{Derka, Martin}
\index{Lopez-Ortiz, Alejandro}
\index{Maftuleac, Daniela}

\begin{document}
\thispagestyle{empty}
\maketitle

\begin{abstract}
We investigate the problem of determining if a given graph corresponds to the dual of a triangulation
of a simple polygon. This is a graph recognition problem, where in our particular case we wish
to recognize a graph which corresponds to the dual of a triangulation of a simple polygon with or without
holes and interior points. We show that the difficulty of this problem depends critically on the amount of information given and
we give a sharp boundary between the various tractable and intractable versions of the problem.
\end{abstract}

\section{Introduction}

Triangulating a polygon is a common preprocessing step for polygon exploration algorithms~\cite{cit:Maf} among many other applications (see~\cite{cit:hjelle}). The exploration of the polygon is thus reduced to a traversal of the triangulation, which is equivalent to a vertex tour of the dual graph of the triangulation.  In the study of lower bounds for such a setting, the question often arises if a given constructed graph is or is not the dual of a triangulation of an actual polygonal region (with or without holes)~\cite{cit:Maf}.
Thus, the recognition of a graph class is a well established problem of theoretical interest and given the importance of triangulations likely to be of use in the future.
%
More formally,
given a graph, does it represent a triangulation dual of a simple polygon?
There are three aspects of this problem: the geometric problem, the topological problem and the combinatorial problem\footnote{In~\cite{cit:Sug2}, Sugihara and Hiroshima call ``the topological embedding problem'' what we call ``the combinatorial problem'' here.}.
In the geometric problem, we are given a precise embedding of the graph.  In the topological problem, we are
given a topological embedding (also called ``face embedding''). In the combinatorial problem, we are given the adjacency matrix only. 
Furthermore, the problem can be stated in both the decision version when the task is to recognize the graph of a triangulation, and the constructive version when the task is to realize the corresponding triangulation.
For some graph classes, recognition may be easier than realization.

Some specialized versions of this problem were studied in the past.
Sugihara, and Hiroshima~\cite{cit:Sug2} as well as Snoeyink and van Kreveld~\cite{cit:Sno} consider the problem of realization of a Delaunay triangulation for the combinatorial version of the problem.
In~\cite{cit:Oka}, the authors define three aspects of the recognition problem of a Voronoi/Delaunay diagram, where the first two of them are what we call the geometric and topological aspects.
The most relevant part of their work is the following question in the geometrical setting~\cite[Problem V10, p.~108]{cit:Oka}: \emph{Given a triangulation graph, decide whether it is a (non-degenerate) Delaunay triangulation realizable graph}. For this case, the authors give necessary and sufficient conditions for a graph to be Delaunay triangulation realizable graph in the geometric setting.

In this paper, we study the problem of recognizing the dual of a triangulation of a simple polygon with or without holes and interior points in the geometric, topological and combinatorial setting.
To the best of our knowledge, this paper is the first work which considers the problem for general triangulations of polygons. We draw a clear line between tractability and NP-completeness of the problem as the degrees
of freedom increase from the geometric to the topological to the combinatorial problem and as we consider holes. Our results are summarized in Table~\ref{tab:results}. The recognition algorithms presented in this paper are constructive and allow realization of the polygon.

\begin{table}
\centering
%
\includegraphics[width=\columnwidth]{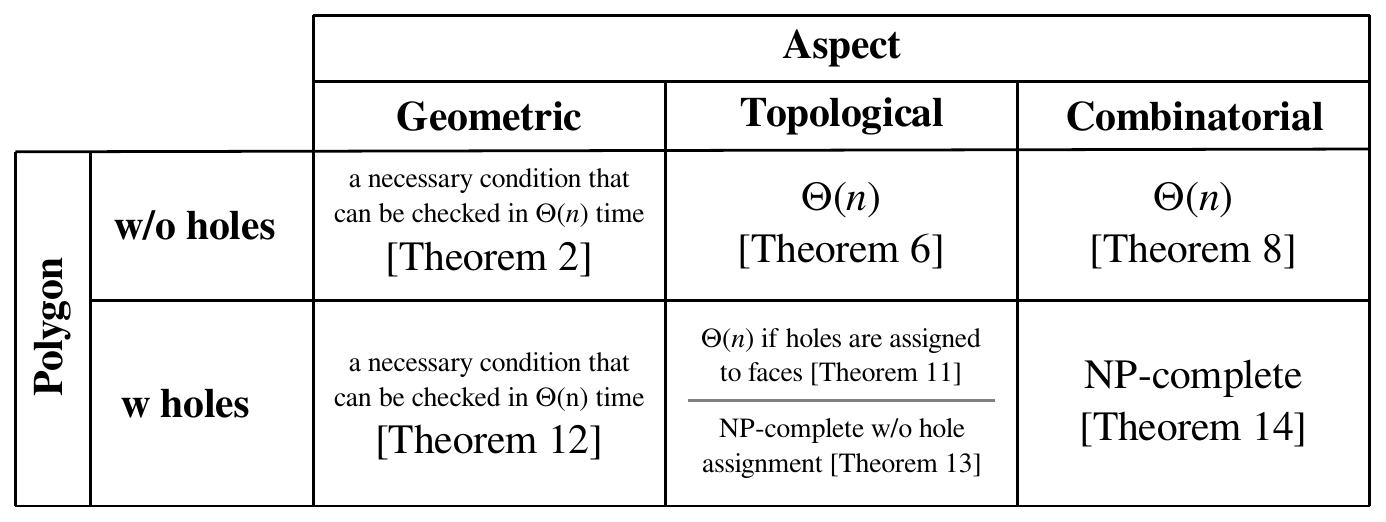}
\caption{Summary of results.}
\label{tab:results}
\end{table}


\section{Preliminaries}
\label{sec:preliminaries}



Let $P$ be a simple polygon with or without holes with $n$ vertices, $S$ a set of $m$ interior points located inside $P$ and $\mathcal T$ a \emph{triangulation} of the $n+m$ given points inside $P$ (for an example of a triangulation, see solid lines in Fig.~\ref{fig:infinity}(a)).
Let $G$ be the \emph{graph of the triangulation} $\mathcal T$ as the graph on vertices $P\cup S$ plus an additional vertex $v$ ``at infinity''
located outside $P$ and the edges of $G$ are the edges of $\mathcal{T}$ plus the edges connecting every vertex on the boundary of $P$ to $v$ (see Fig.~\ref{fig:infinity}).

\begin{figure}[h]
\centering
\includegraphics[width=.6\textwidth]{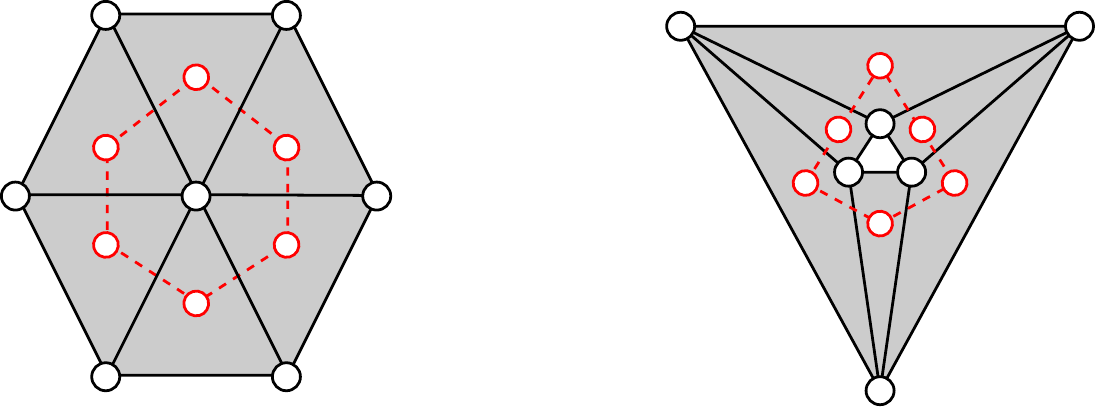}
\caption{Two triangulations of a polygon with isomorphic dual graphs if the point at infinity is omitted.}
\label{fig:notunique}
\end{figure}

This paper reconstructs triangulations of polygons from duals via
reconstructing their graphs (which include the point at infinity). As we show, the point
at infinity provides one with tools which are fully sufficient for such a reconstruction. If
graphs of triangulations were defined without points at infinity, one would discover that
there are many triangulations of a polygon with the same dual (see Fig.~\ref{fig:notunique}).
Furthermore, we suggest that adding the point at infinity to representations
of triangulations is easy to accomplish: Given a triangulation $\mathcal{T}$ of a polygon, one can
construct its graph $G$ by adding the point at infinity. In the other direction, if the vertex at infinity
is known, one can easily construct triangulation $\mathcal{T}$ from its graph $G$. The information about
which is the point at infinity can be given as a part of the input, or in some cases, this may be even
implicitly determined by formulation of the problem (see Definition~\ref{def:problem}; TDR-without-holes).

\begin{figure}
\begin{tabular}{cccc}
\includegraphics[width=.4\textwidth]{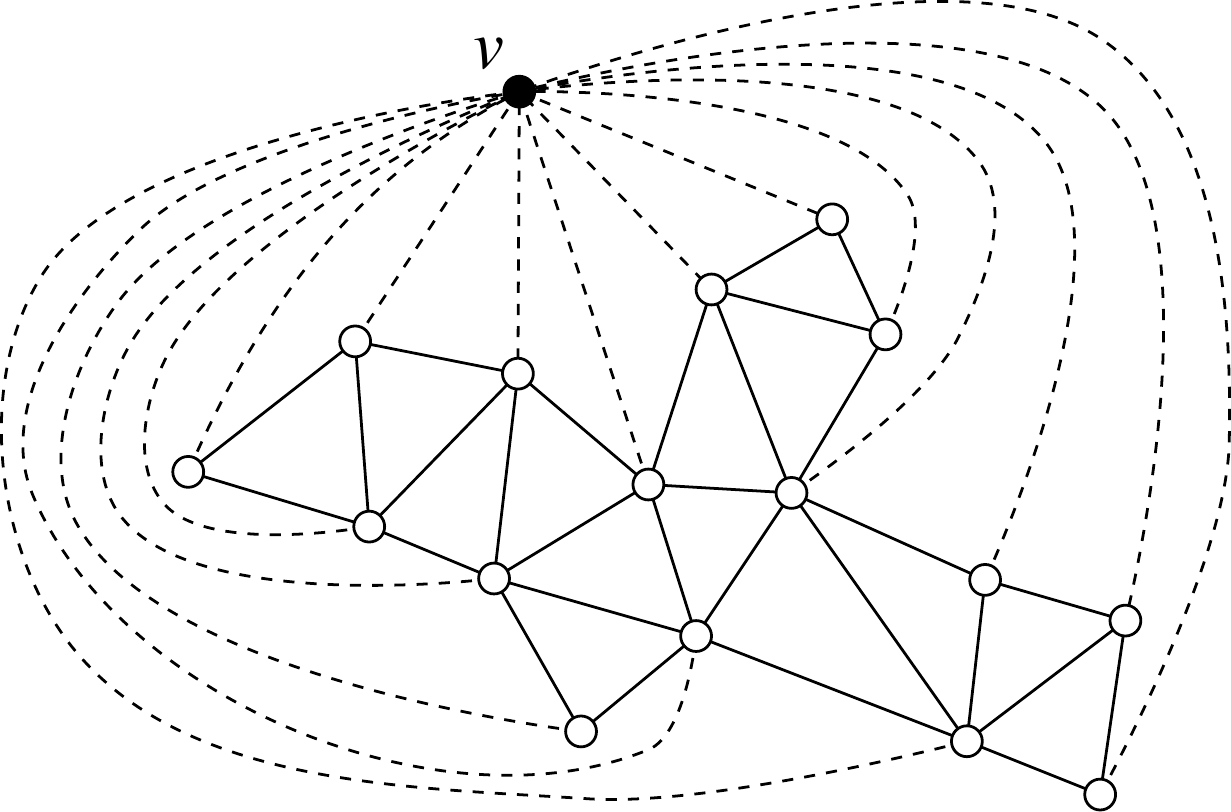}
   &
   ~~~~~~~~~~~~~
   &
\includegraphics[width=.4\textwidth]{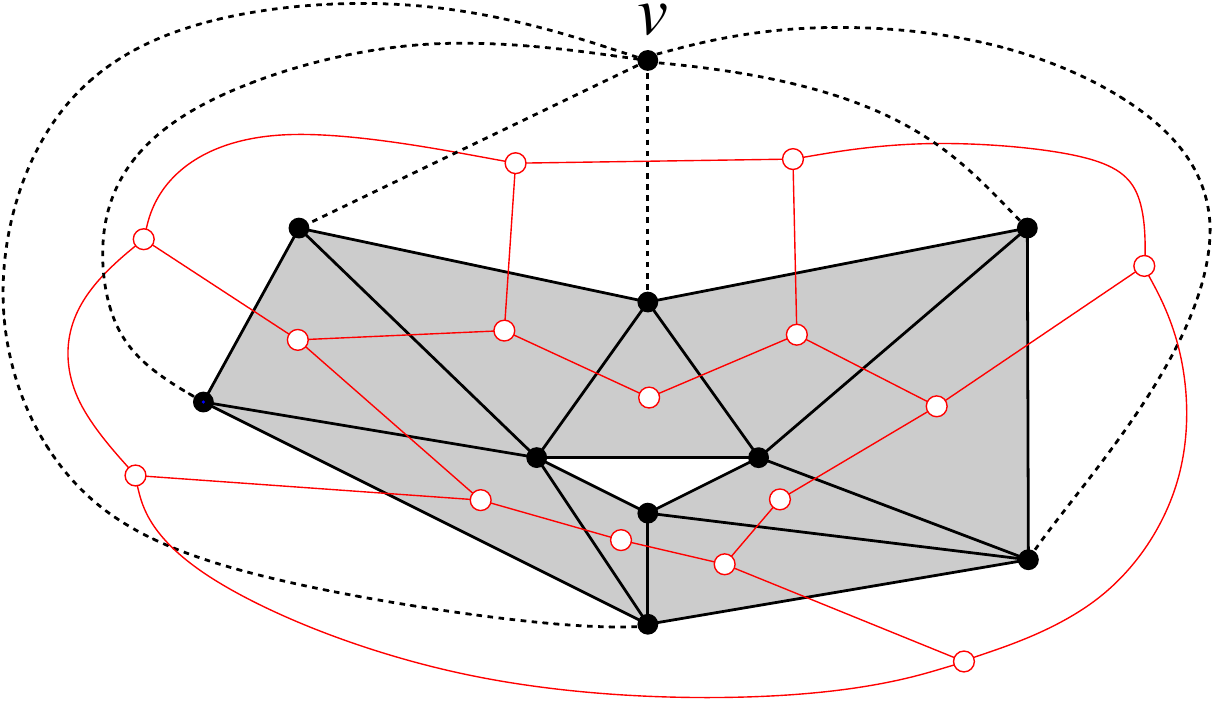}
   \\
   (a)
   &
   ~~~~~~~~~~~~~
   &
   (b)
\end{tabular}
\caption{(a) An example of a triangulation of a polygon (solid lines) and its graph (solid and dashed lines),
(b) a polygonal region $P$ with one (white) hole
(shown in solid black lines and gray interior);
its triangulation $\mathcal T$ (in solid black lines);
the graph $G$ of the triangulation $\mathcal T$ (in black, solid and dashed lines);
the dual graph $G^*$ of the triangulation (in solid red lines).}
\label{fig:infinity}
\end{figure}

Given a plane graph $\Gamma$,
the \emph{dual graph} of $\Gamma$, denoted by $\Gamma^*$, is a planar graph whose vertex set is formed by the faces of $\Gamma$ (including the outer face), and two vertices in $\Gamma^*$ are adjacent if and only if the corresponding faces in $\Gamma$ share an edge. 

\noindent
Let $G$ be a graph of a triangulation of a polygon $P$ and $G^*$ its the dual graph. For brevity, we say that $G^*$ is the
\emph{dual graph of the triangulation} $\mathcal T$ and from now on we will use this notion instead of ``the dual graph of the graph of a triangulation $\mathcal T$.''

\begin{definition}[The TDR Problems]
\label{def:problem}
Given a planar graph $G^*$, decide if
$G^*$ is a dual graph of a triangulation of a polygon $P$ with a set of interior points $S$. We distinguish between:
(1) \emph{TDR-without-holes} if $P$ is not allowed to have holes and $S = \emptyset$;
(2) \emph{TDRS-without-holes} if $P$ is not allowed to have holes and $S$ may be non-empty;
(3) \emph{TDR-with-known-holes} if $P$ is allowed to have holes, $S = \emptyset$, and the positions of holes are part of the input; and
(4) \emph{TDR-with-unknown-holes} if $P$ is allowed to have holes, $S = \emptyset$, and the positions of holes are unknown.
\end{definition}

The following proposition summarizes some well-known facts about planar graphs and their duals.
\begin{proposition}
\label{prop:simple-properties}
1.~The dual of a planar graph $G$ is a planar graph.
2.~The embedding of a $3$-connected graph is unique up to the choice of the outer face.
3.~The dual graph of a $3$-connected planar graph is a $3$-connected planar graph.
\end{proposition}
\begin{proof}
\noindent (1) Consider a planar embedding of $G$. Every vertex of $G^*$ can be embedded inside a face that it represents and connected to any point on the boundary of the face by a ``half-edge'' without introducing any crossings. By joining the ``half-edges'', one can construct a planar embedding of $G^*$.

\noindent (2) is a well-known theorem of Whitney; see e.g.~\cite{cit:diestel} for the proof.

\noindent (3) is another well-known fact. A quick argument can be given using Steinitz's theorem~\cite{cit:grun}.
A planar 3-connected graph $G$ can be realized as a polyhedron $P$. Take its dual polyhedron $P'$, whose graph is $G^*$, i.e., the dual graph to $G$. Using Steinitz's theorem again, $G^*$ is planar and 3-connected.\qed
\end{proof}

\section{Triangulation Dual Recognition (TDR)}
\label{sec:recognition}

We present a sequence of increasingly complex dual recognition problems. We draw a clear line between the tractability of the problem and the NP-completeness depending on the degrees of freedom in the particular setting being considered.
We first establish some properties of the triangulation dual of a polygon that will
allow us to decide if the input graph is a dual of a triangulation or not.
We consider separately the cases where the triangulated polygon has holes or not,
and contains interior points or not.

We consider three aspects of this problem depending on the amount of information given. In the most restricted case, we are given a \emph{geometric} embedding of the dual of a triangulation. Each
triangle of $G$ is represented in the dual $G^*$ by a distinguished point in its interior. In particular, following Hartvigsen \cite{cit:hart} we consider the circumcenter of the triangle 
(which does not necessarily lie inside the triangle) and we are given the edge adjacencies between the triangles. In the second
case we are given the faces of the dual of the triangulation but not their precise geometric embedding. This forms the \emph{topological} recognition problem. Lastly, in the least restrictive case we are simply given a dual graph without any knowledge of which vertices form a face in the triangulation dual. This is the \emph{combinatorial} recognition problem.

\paragraph{Geometric TDR- and TDRS-without-holes.}\label{subs:gtdrwoh}

For the geometric recognition problem, we do not consider the point at infinity, since it does not have a natural geometric representation. Thus, in this problem the input is a geometric embedding of the dual of the triangulation
$\mathcal T$. In the dual, each triangle is represented by a distinguished point. The natural choices for such a point are (a) the circumcenter, (b) the incenter, (c) the orthocenter, (d) the centroid or (e) an arbitrary point in the interior of the triangle.

For the case of (a), the circumcenter, which is the choice of Hartvigsen for the recognition of Delaunay triangulations~\cite{cit:hart}, we use a similar technique and create a two dimensional linear program. This is based on the observation that the edges in the triangulation are perpendicular to the dual edges in the geometric embedding. The intersections of such edges are the vertices of the polygon. Observe as well that the center of the triangulation edges lies on the corresponding dual edge. We can then set up a linear program with the coordinates of the vertices of the polygon as unknowns, and the orthogonality and bisection equations as linear constraints. We then solve the two dimensional LP program in linear time using Megiddo's fixed dimension LP algorithm~\cite{cit:megi}. If there
is no feasible solution then we know that necessarily the given input graph is not the dual of a triangulation
since otherwise, the actual triangulation graph satisfies the given linear constraints.

A similar approach works for the case of (c) the centroid. The location of the vertices and the median points are the unknowns and the collinearity with the centroid is expressed as a convex combination of those two vertices
with the centroid trisecting the line segment (i.e. $\lambda=1/3$ in the convex combination equation). This produces a set of linear equations which can also be addressed using an LP solver.

\begin{figure}[h]\centering
\begin{tabular}{cc}
   \includegraphics[width=0.75\textwidth]{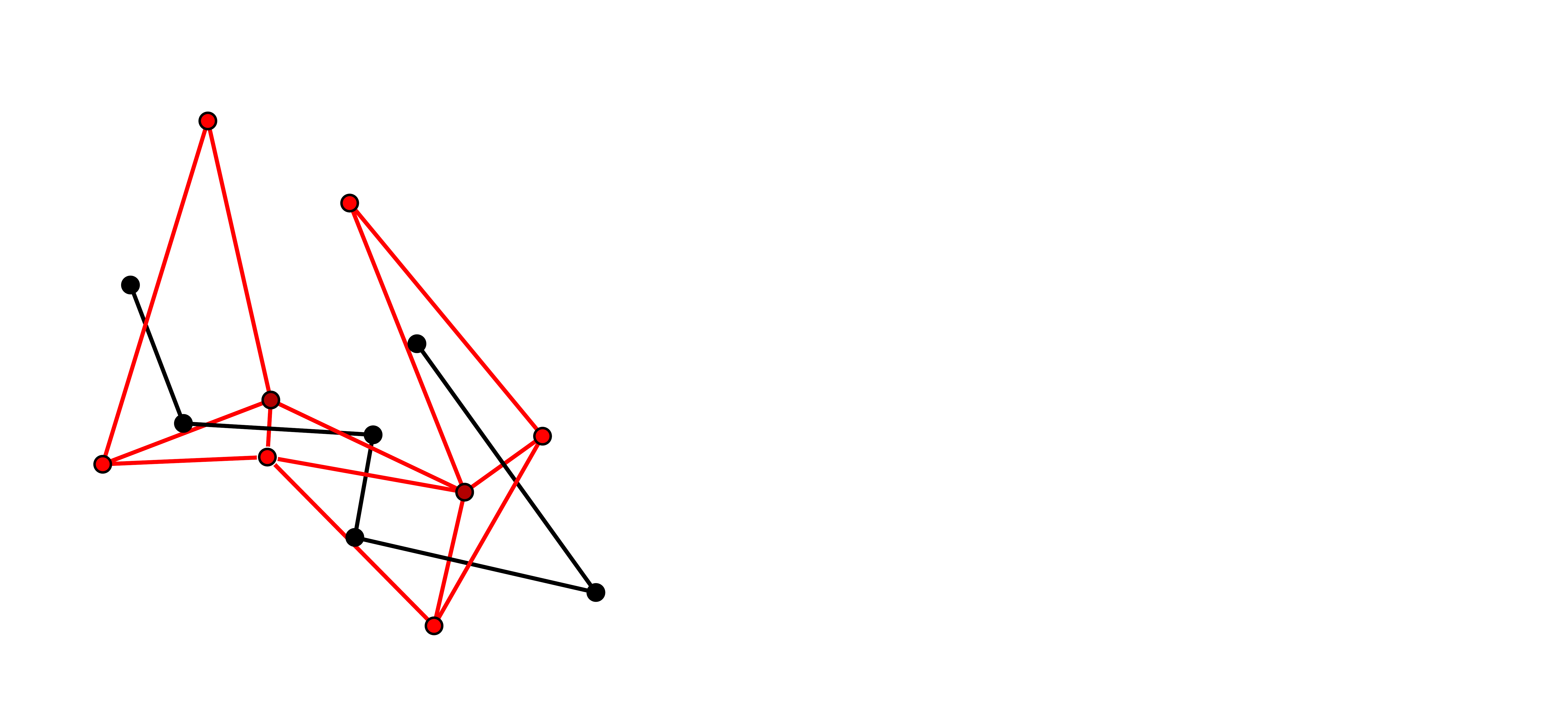}\hspace*{-5cm}
   &
   \includegraphics[width=0.75\textwidth]{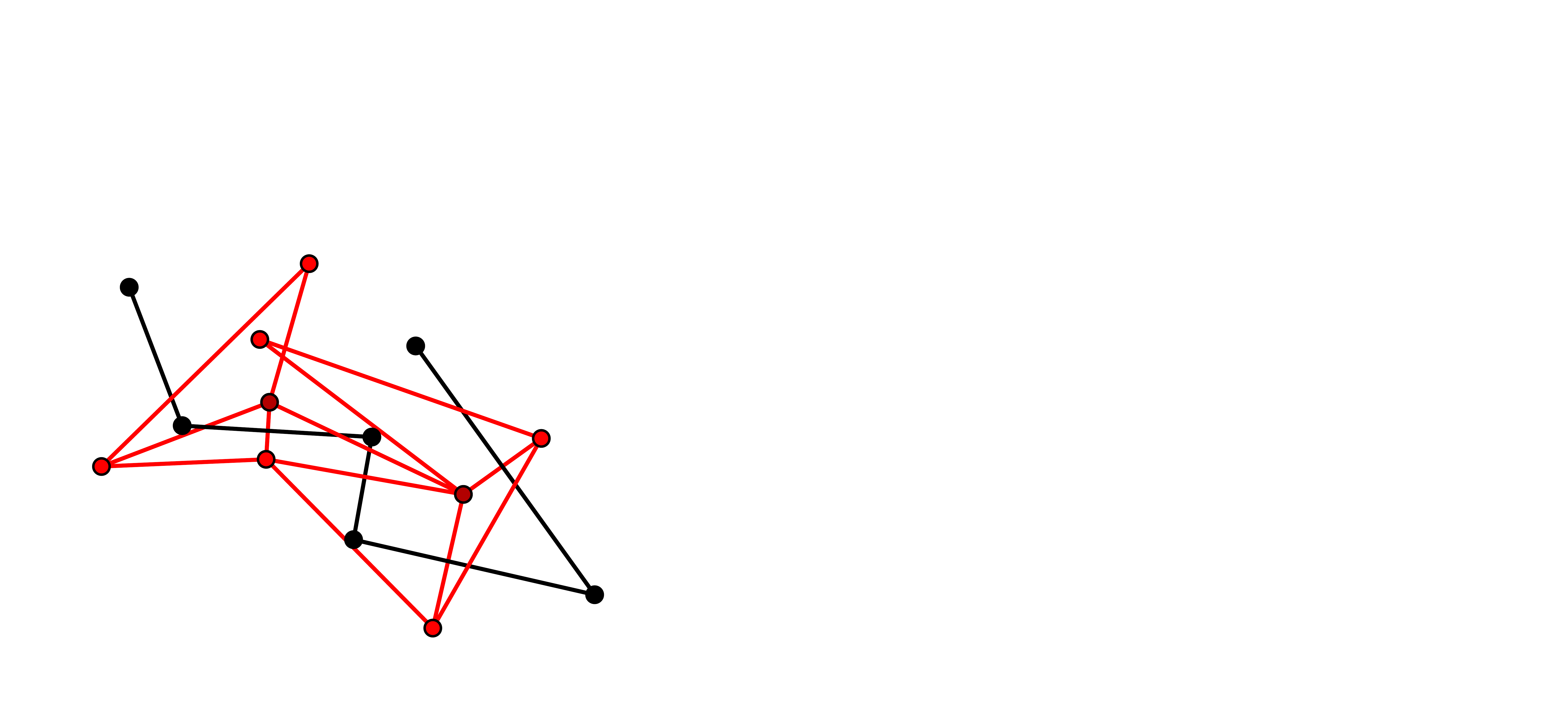}\hspace*{-4cm}
   \\
   (a)
   &
   (b)
\end{tabular}
\caption{The figure illustrates a dual graph in black with the vertices representing circumcenters of triangles of a triangulation. The red edges show a candidate triangulation graph constructed from the LP-solver. Diagram (a) illustrates a valid triangulation and (b) an invalid triangulation that violates planarity.} \label{fig:circumcentres}
\end{figure}

\begin{figure}[h]\centering
\begin{tabular}{ccc}
   \includegraphics[width=0.5\textwidth]{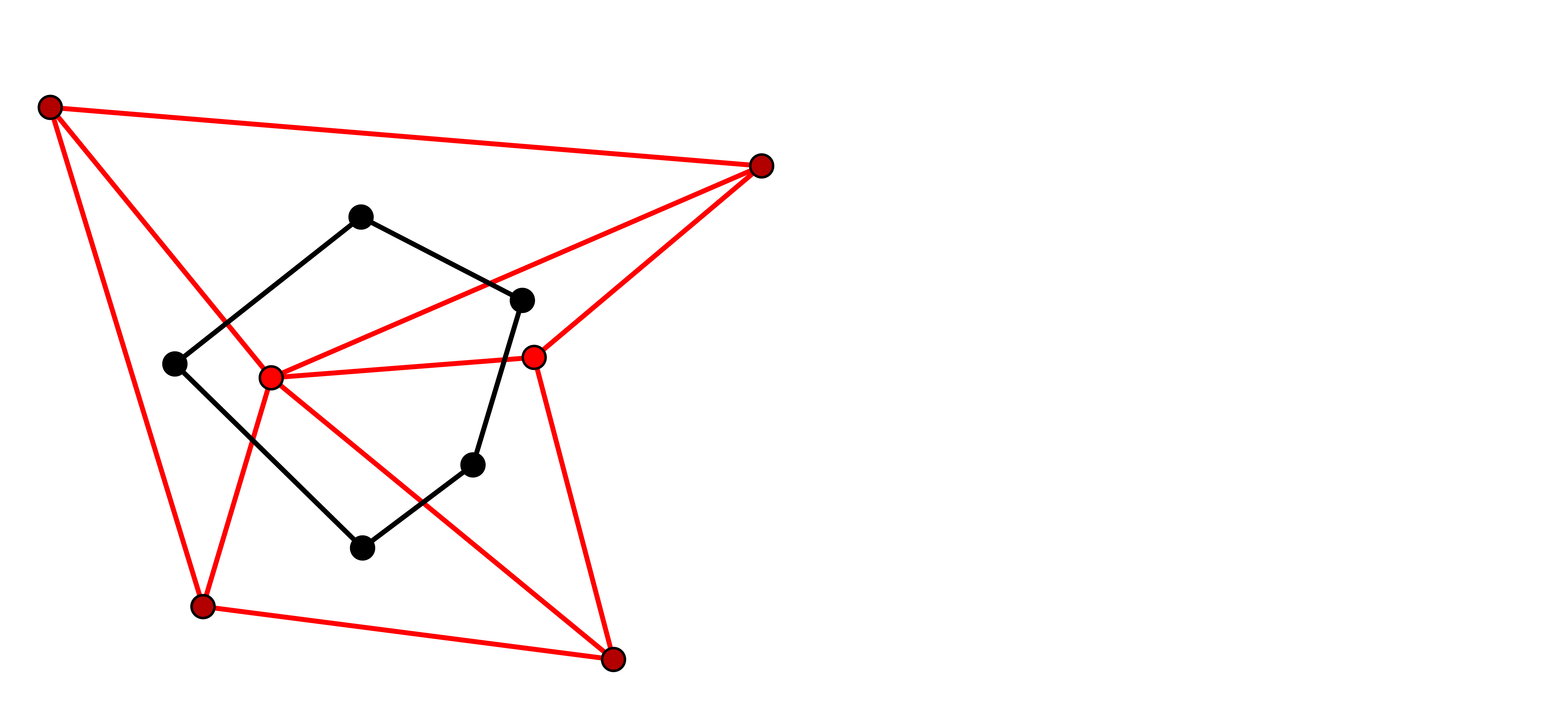} \hspace*{-3cm}
   &
   \includegraphics[width=0.5\textwidth]{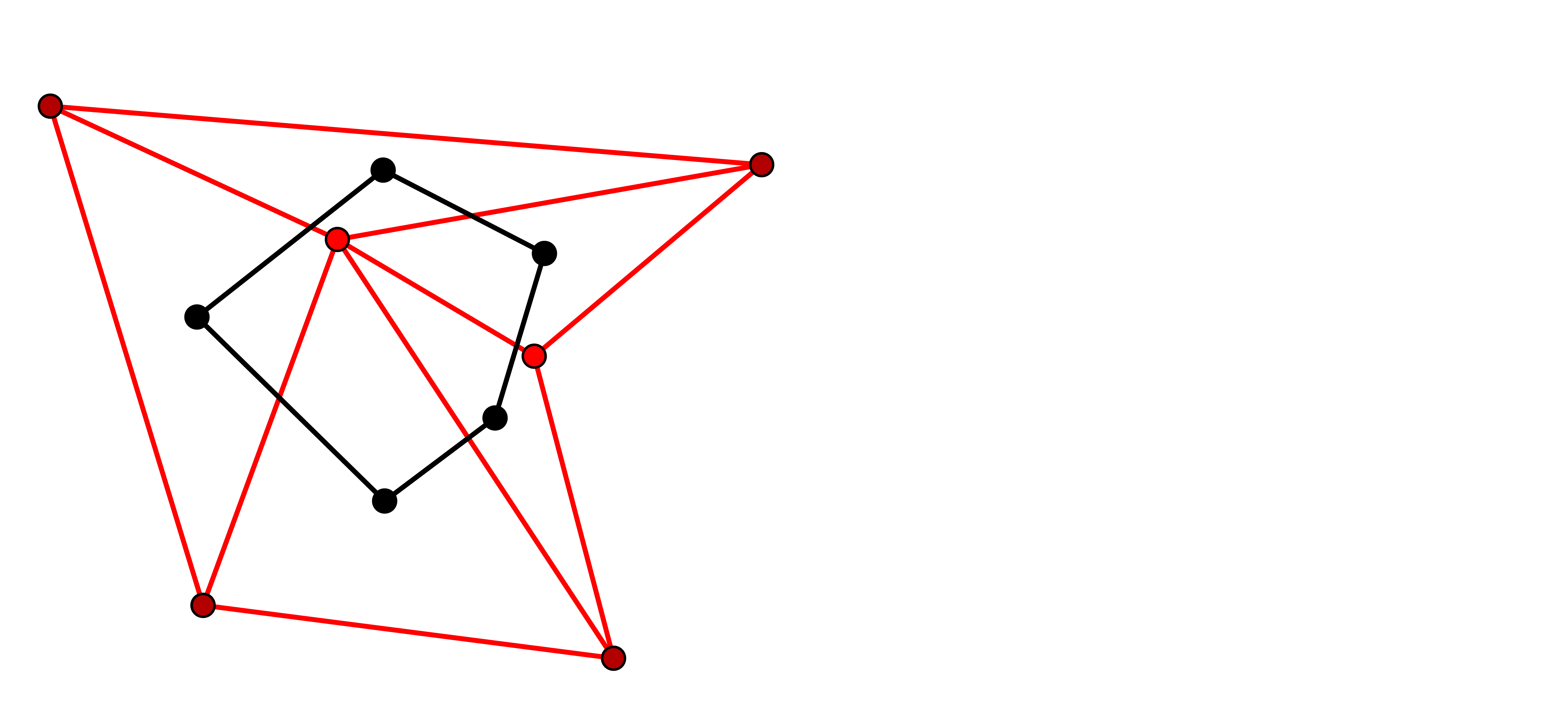} \hspace*{-3cm}
   &
   \includegraphics[width=0.5\textwidth]{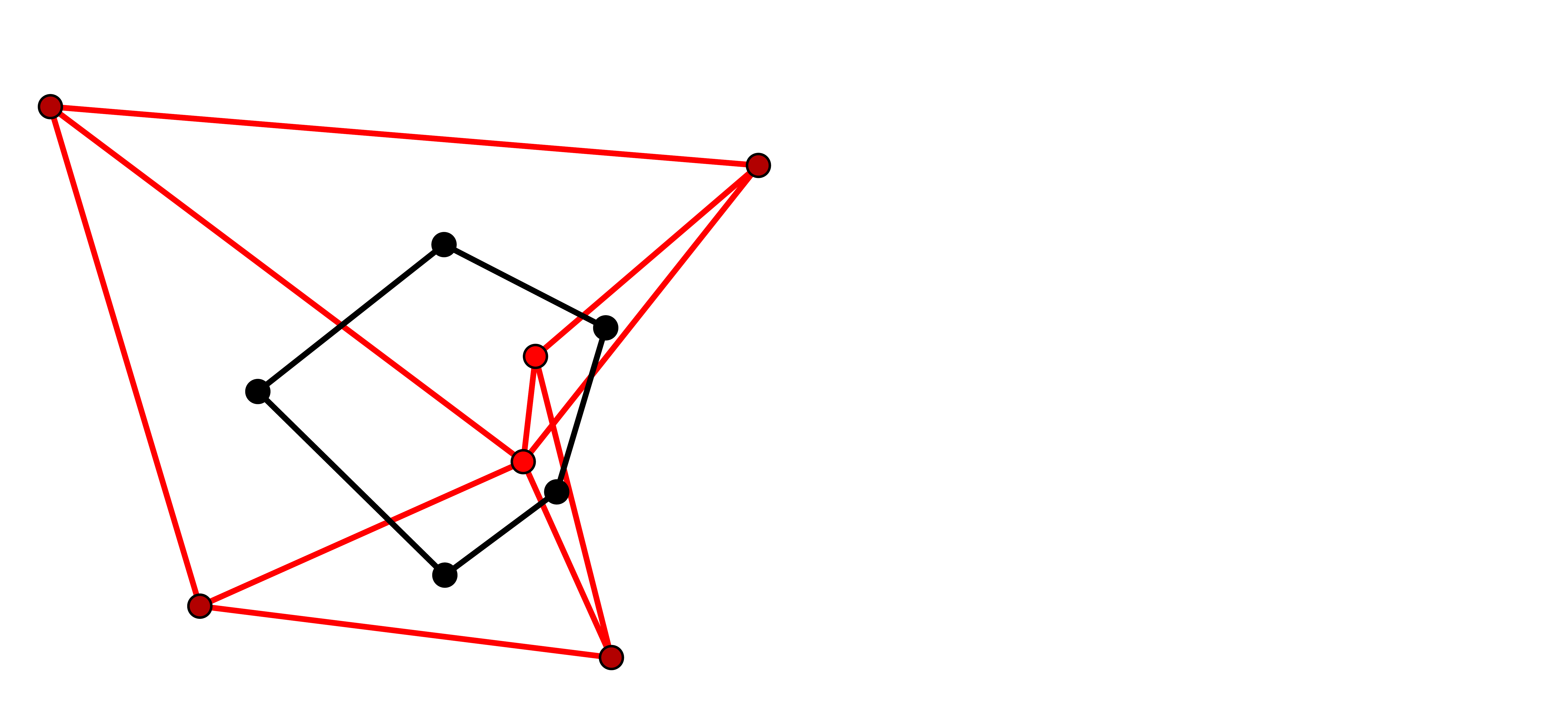} \hspace*{-3cm}
   \\

   (a)
   &
   (b)
   &
   (c)
\end{tabular}
\caption{The figure illustrates a dual graph in black with the vertices representing centroids of triangles of a triangulation. The red edges show a candidate triangulation graph constructed from the LP-solver. Diagrams (a) and (b) illustrate two valid triangulations and (c) an invalid triangulation that violates planarity.} \label{fig:centroids}
\end{figure}

\begin{theorem}\label{thm:circumcentres}
The linear program described above gives a necessary condition for the realization of the geometric TDR- and TDRS-without-holes problems in linear time with input $G^*$ given the triangulation graph with the circumcenters/centroids of the triangles of $G^*$ as vertices.
\end{theorem}

However, it is important to observe that the feasible solution by the LP only obeys orthogonality/median constraints and has no knowledge of planarity constraints of the resulting triangulation. Thus, the proposed solution might not be a realizable triangulation.
One way of resolving this problem is testing (in linear time) if the proposed solution is planar. 
If it is, we now have a realization of the triangulation. If on the other hand the solution is not planar, we cannot decide if there was not possibly
another realization that would have been planar. This is illustrated in Fig.~\ref{fig:circumcentres} and \ref{fig:centroids}, where we give different solutions to the LP constraints over the same dual triangulation graph, one leading to a feasible triangulation and the other does not.

It remains an open problem if recognition is possible under either of this models, as well as any bounds for necessary and/or
sufficient conditions under other choices for triangle representatives.

To the best of our knowledge, planarity constraints between two triangles are a disjunction of three linear constraints which leads to a third degree equation which as such cannot be resolved using the LP program. Thus full recognition of geometric graphs remains an open problem.

\paragraph{Topological TDR- and TDRS-without-holes.}
\label{subs:ttdrwoh}

There are two cases of the problem in this setting: (1) the output triangulation possibly contains interior points (TDRS-without-holes) and (2) the triangulation does not contain any interior points (TDR-without-holes).

\paragraph{TDRS-without-holes.}
\label{subsec:ttdrs-without-holes}
Let us begin with the proof of the following lemma:

\begin{lemma}
\label{lem:ttdrs-3-con}
Let $P$ be a polygon without holes, $S$ be a set of points in the interior of $P$, and $\mathcal{T}$ a triangulation of $P \cup S$. The graph $G$ of $\mathcal{T}$ is a $3$-connected maximal planar graph.
\end{lemma}
\begin{proof}
By definition, $\mathcal{T}$ can be drawn in the plane without crossing edges. Hence, the graph (let us call it $T$) induced by the vertices of $\mathcal{T}$ is planar. Since $v$ is located outside the polygon $P$ and is not connected to any vertex in the interior of $P$, the graph $G$ is planar. As the boundary of every face in the subgraph $T$ of $G$ induced by the vertices of $\mathcal{T}$ is a simple cycle, $T$ is $2$-connected. Furthermore, every $2$-cut in $T$ is formed by the vertices on the boundary of $\mathcal{T}$. 
Hence, by adding $v$ to $T$, we obtain a $3$-connected graph $G$. This graph is maximal planar as every face (including the outer face) is a triangle. \qed
\end{proof}

We establish necessary (Lemma~\ref{lem:simple-properties}) and sufficient (Lemma~\ref{lem:dual-planar}) conditions for a graph to be a dual graph of a triangulation of a polygon with no holes.
\begin{lemma} \label{lem:simple-properties}
If $G^*$ is a dual graph of a triangulation of a polygon $P$ and set $S$ of interior points inside $P$ with no holes, then $G^*$ is a planar $3$-regular $3$-connected graph.
\end{lemma}
\begin{proof}
The fact that $G^*$ is planar and $3$-connected follows from Lemma~\ref{lem:ttdrs-3-con} and Proposition~\ref{prop:simple-properties}. As the graph $G$ of the triangulation is a maximal planar graph, every face of $G$ is a triangle. Hence, every vertex in $G^*$ has precisely three incident edges.
\end{proof}

\begin{lemma} \label{lem:dual-planar}
If $G^*$ is planar, $3$-regular and $3$-connected, then $G^*$ is a dual graph of a triangulation of a polygon without holes $P$ and a set $S$ of interior points.
\end{lemma}
\begin{proof}
By Proposition~\ref{prop:simple-properties}(3), $G^*$ has a dual graph $G$ which is $3$-connected, and thus $G$ can be uniquely embedded in the plane up to the selection of the outer face (Proposition~\ref{prop:simple-properties}(2)). Such an embedding can be achieved using straight lines only (see e.g.~\cite{cit:deFra}). Now, remove the vertex $v$ of $G$ which represents the outer face of $G^*$. As $G$ is $3$-connected, by removing $v$, we obtain a $2$-connected plane graph $G'$. Hence, every face of $G'$ is a simple cycle, and it is a triangulation of a polygon formed by its outer face. Moreover, since the graph $G^*$ is 2-connected and every face is a triangle, it is the triangulation of a polygon.
\end{proof}

\begin{theorem}\label{thm:tdrs-no-holes}
The answer to the topological TDRS-without-holes problem is ``yes'' if and only if the input $G^*$ is a $3$-connected $3$-regular planar graph. Furthermore, such a polygon can be constructed in linear time. 
\end{theorem}
\begin{proof}
The first part of the claim follows directly from Lemmas~\ref{lem:simple-properties} and~\ref{lem:dual-planar}. The linear running time follows from linearity of 
verifying $3$-connectivity of a graph~\cite{cit:hopcroft}. The reconstruction is linear as the number of faces in any planar graph is linear due to Euler's formula, so the dual graph can be constructed in linear time. The straight line embedding can be found in linear time as well~\cite{cit:deFra}, and deleting a vertex from an embedded graph takes at most $\mathcal{O}(n)$ steps too.
\end{proof}

\paragraph{TDR-without-holes.}
\label{subsec:ttdr-without-holes}

Previously, we showed that topological TDRS-without-holes problem can be solved in linear time. This can be done even if the set of interior points $S$ is required to be empty (i.e., the graph of the triangulation should consist only of vertices at the boundary of the polygon $P$, and one vertex outside $P$). 

\begin{proposition}\label{lem:pol-no-holes-Steiner}
Let $G^*$ be a $3$-regular planar graph and $\widetilde{G^*}$ the subgraph of $G^*$ obtained by removing the vertices of the outer face.
$\widetilde{G^*}$ is a tree if and only if it corresponds to a polygon with no holes or interior points.
\end{proposition}
\begin{proof}
All the outer vertices of $G^*$ correspond to faces introduced by the point at infinity vertex and all the interior vertices of $G^*$ correspond to the faces
of the triangulation of a polygon.

\noindent $(\Leftarrow)$
The dual graph of a triangulation of a polygon without holes or interior points is a tree~\cite{cit:deBerg}~(p. 48).

\noindent $(\Rightarrow)$
It is a simple exercise to show by induction that every tree of degree at most 3 is the dual of a triangulation of a polygon without holes.
\qed
\end{proof}

\paragraph{Combinatorial TDR- and TDRS-without-holes.}\label{subs:ctdrwoh}


It is easy to see that the topological and combinatorial input are equivalent in this case. 
Since $G^*$ must be $3$-connected and $3$-regular, the algorithm can first verify this necessary condition. If it is satisfied, it can construct the embedding (e.g., applying the linear straight line embedding algorithm of~\cite{cit:deFra}) and proceed with the topological input.

\begin{theorem}\label{thm:comb-tdr-no-holes}
The answer to combinatorial TDR- and TDRS-without-holes problems is affirmative if and only if the input $G^*$ is a $3$-connected $3$-regular planar graph.
Furthermore, such a polygon can be constructed in linear time.
\end{theorem}


\paragraph{Topological TDR- and TDRS-with-known-holes.}\label{subs:ttdrwh}

Let us start with the following observation:

\begin{proposition}
\label{prop:degree2}
If a polygon has a hole, the dual graph $G^*$ of its triangulation contains vertices of degree $2$ or less. 
\end{proposition}
\begin{proof}
If we triangulate the polygon together with its holes (treating the vertices at the boundary of the hole as interior points), we obtain a 3-regular 3-connected dual graph. We now construct the dual graph $G^*$ of the triangulation and remove the vertices that correspond to the faces inside the holes of the polygon. The remaining graph is connected and has at least one vertex of degree 2 or less.
\qed
\end{proof}

From the proof of Proposition~\ref{prop:degree2}, we can see that vertices of degree 2 in $G^*$ are adjacent to holes in the initial polygon.
Observe that if $P$ is a polygon with or without holes, the triangles of the graph $G$ of a triangulation created by the point at infinity and the outer face of the polygon form a 3-connected graph. Thus, the dual graph $G^*$ cannot contain a 2-cut on the outer face corresponding to these triangles.

We can associate each degree 2 vertex to its adjacent hole. Formally, let $G^*$ be a planar graph that contains at least one vertex of degree 2. We define an \emph{assignment} of a vertex $u$ of degree $2$ to a face of $G^*$, as a mapping $\mathcal H$ from the set containing $u$ to the set of faces incident to $u$ of $G^*$, such that if $u$ is incident to faces $F$ and $F'$ in $G^*$, then 
 $\mathcal{H}(u) \in \{F, F'\}$.
The same way we can define: an \emph{empty assignment}, which does not assign any vertex of degree 2 to a face of $G^*$; a \emph{partial assignment}, which assigns a subset of vertices of degree 2 to their incident faces in $G^*$ and a \emph{total assignment} which assigns all the vertices of degree 2 to faces of $G^*$ (see Fig.~\ref{fig:assignment} for a total assignment example).

\begin{lemma} \label{lem:deg2=hole}
Let $\{G^*,\mathcal H\}$ be such that $G^*$ is a dual graph of a triangulation of a polygon with holes. A face that is assigned vertices of degree $2$ contains a hole in the initial polygon. Moreover, $\mathcal H$ assigns to each face of $G^*$ zero or at least three vertices of degree $2$.
\end{lemma}
\begin{proof}
Since $G^*$ is a dual graph of a triangulation of a polygon with holes, the vertices on the outer face have degree 3.
Let $u^*$ be a vertex of degree 2, thus $u^*$ is an interior vertex. Let the assignment $\mathcal H$ assign $u^*$ to a face containing~$u^*$.
Recall that every vertex $v^*$ in $G^*$ corresponds to a face $F$ in the initial triangulation graph $G$, every face $F^*$ in  $G^*$ corresponds to a vertex $v$ in $G$ and every edge $e^*$ in $G^*$ corresponds to an edge $e$ crossing $e^*$ in $G$.
Since $u^*$ in $G^*$ has the degree 2, then only two different (topological) scenarios are possible: the face, which is a triangle, corresponding to $u^*$ in $G$ has one vertex in one face of $G^*$ containing $u^*$ and two vertices in the other face of $G^*$ containing $u^*$ or the other way around (see Fig. \ref{fig:2cases}).

\begin{figure}[h]\centering
\begin{tabular}{cccc}
   \includegraphics[width=0.35\textwidth]{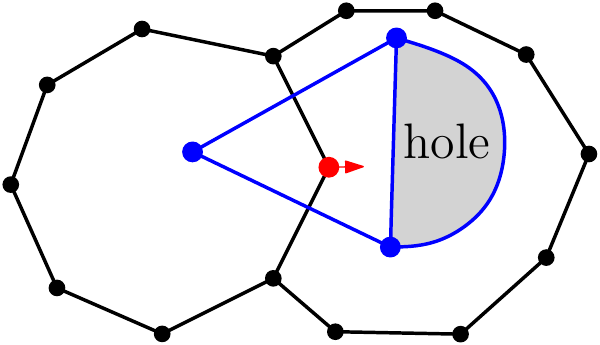}\hspace{0.3cm}
   &
   \includegraphics[width=0.35\textwidth]{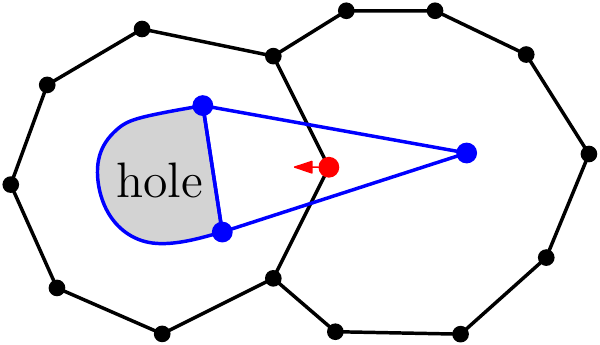}\hspace{0.3cm}
   \\
   (a)\hspace{0.3cm}
   &
   (b)\hspace{0.3cm}
\end{tabular}
\caption{Two possible cases of reconstruction of the initial graph of the triangulation.} \label{fig:2cases}
\end{figure}

Which of these two cases is the right one is given by the assignment $\mathcal H$ by assigning $u^*$ to exactly one of the two faces containing $u^*$. If $u^*$ is assigned to a face in $G^*$ then in $G$ the triangle $U$ corresponding to $u^*$ has the two vertices in that face of $G^*$.
Since we know that every face of $G^*$ corresponds to one vertex in $G$ and in our case the face of $G^*$ to which was assigned $u^*$ has two vertices of $G$, this implies that this face contains a hole in the initial polygon.

A hole has length at least 3. Since one vertex of degree 2 assigned to a face induces one edge of the hole, then for $G^*$ to be the dual graph of a triangulation of a polygon with holes, $F$ needs to be assigned at least three vertices of degree 2.
\qed
\end{proof}


We know that the presence of a vertex of degree 2 in the graph $G^*$ means there is a hole in the output polygon in one of the faces incident to this vertex in $G^*$. The reason to define an assignment of vertices of degree 2 to faces of the graph is to establish in which of the two incident faces the hole is contained (see Fig.~\ref{fig:2cases}).
We call $\mathcal H$ a \emph{valid assignment} if we can realize $G^*$ as a triangulation dual of a polygon with holes.
A polygon $P$ with holes is a realization of $\{G^*,\mathcal H\}$ if the polygon is a realization of $G^*$ and $\mathcal H$ is a valid assignment with respect to~$P$.

Let us now focus on the one-edge cuts in $G^*$, such as the one shown in Fig.~\ref{fig:edge-cuts}(a). These one-edge cuts in the dual represent edges in the original polygon where a straight line cut applied to that
edge would separate the polygon into two disjoint subpolygons. The two possible cases of how this separation looks like are illustrated in Fig.~\ref{fig:edge-cuts}(b) and (c).

\begin{figure}[h]\centering
\begin{tabular}{ccc}
\includegraphics[width=0.3\textwidth]{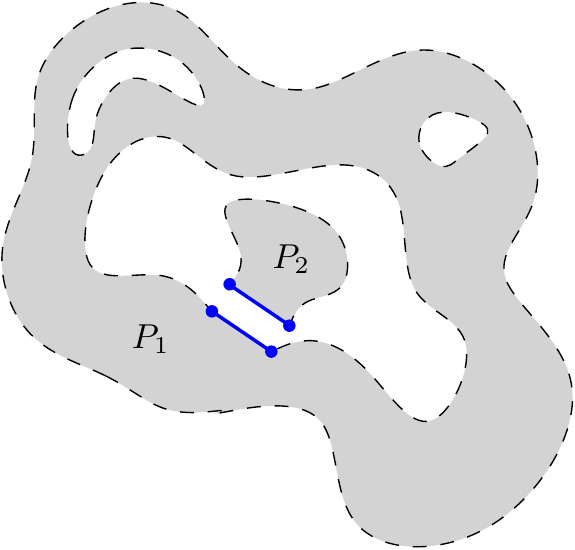}
&
~~~~~~~~~~~~~~
&
\includegraphics[width=0.4\textwidth]{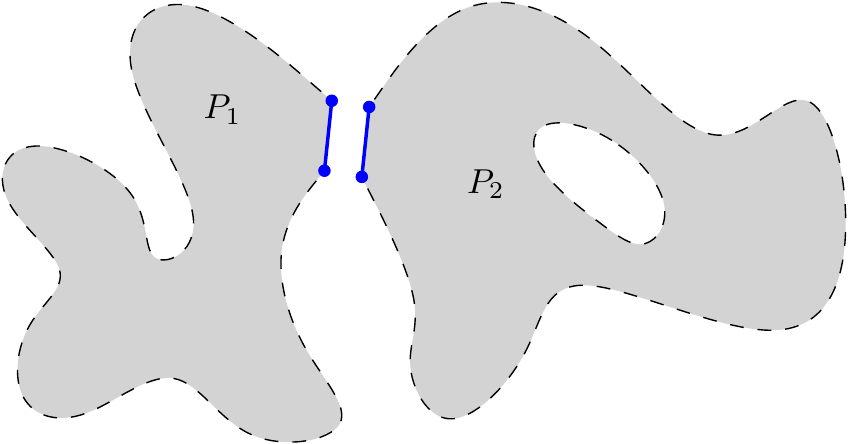} \\
(a)
&
~~~~~~~~~~~~~~
&
(b)
\end{tabular}
\caption{Two possible realizations of two components of an one-edge cut of $G^*$.}
\label{fig:edge-cuts}
\end{figure}

Now we observe that when the first such cut is applied, the case shown in Fig.~\ref{fig:edge-cuts}(c) is not possible since the outer face is 3-connected due to the point at infinity and the upper and lower chain of the original polygon. So in what follows, we need only consider case (b) in the figure.
Let $G_1^*$ and $G_2^*$ denote the subgraph duals of $P_1$ and $P_2$, respectively in $G^*$.
The algorithm now recursively creates a topological embedding for $P_1$ and $P_2$ and merges the two embeddings.
We can show that this process is deterministic and results in a unique topological graph which can be embedded using straight line edges.
This resulting polygonal graph is a simple polygon with point at infinity if and only if $G^*$ is a triangulation dual of a simple polygon.
Hence, we have the following theorem.

\begin{theorem}\label{thm:tdrwkh}
Given an input $\{G^*, \mathcal{H}\}$, the topological TDR- and TDRS-with-known-holes problems are decidable in linear time.
\end{theorem}
\begin{proof}
Recall that the algorithm partitions the dual graph $G^*$ along a one-edge cut. In general, the algorithm processes each of the one-edge cuts in order starting from minimal edge cuts, i.e. cuts in which at least one of the disjoint resulting components has no other one-edge cut.
Without loss of generality we denote by $P_2$ the polygon with no further one-edge cuts and $P_1$ the other component as shown in Fig.~\ref{fig:edge-cuts}(a). Let $G_1^*$ and $G_2^*$ denote the subgraph duals
of $P_1$ and $P_2$, respectively in $G^*$.  Let $F_1$ be the face in $G_1^*$ that contains $G_2^*$. We now consider the topological subgraphs $G_2^* \cup F_1$ and $G_1^*$.

\begin{figure}\centering
\begin{tabular}{ccccc}
\includegraphics[width=0.3\textwidth]{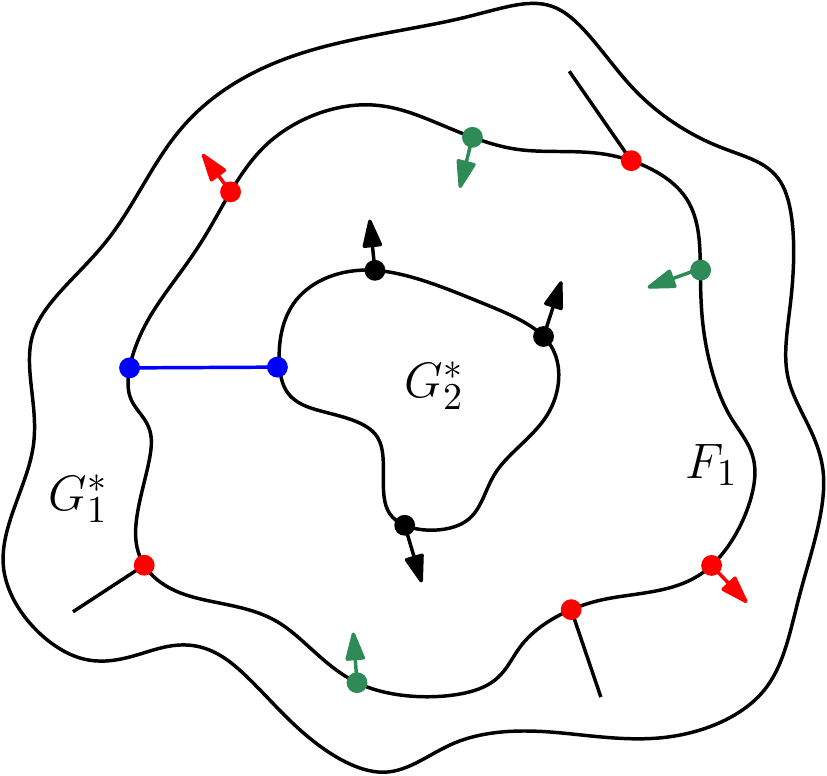}
&
~~~~
&
\includegraphics[width=0.3\textwidth]{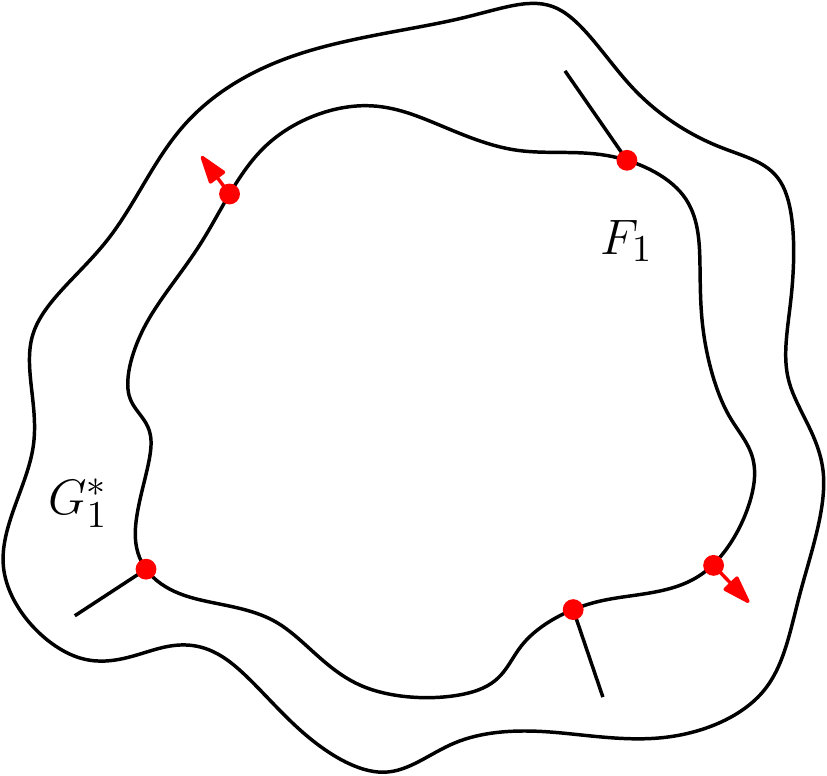}
&
~~~~
&
\includegraphics[width=0.24\textwidth]{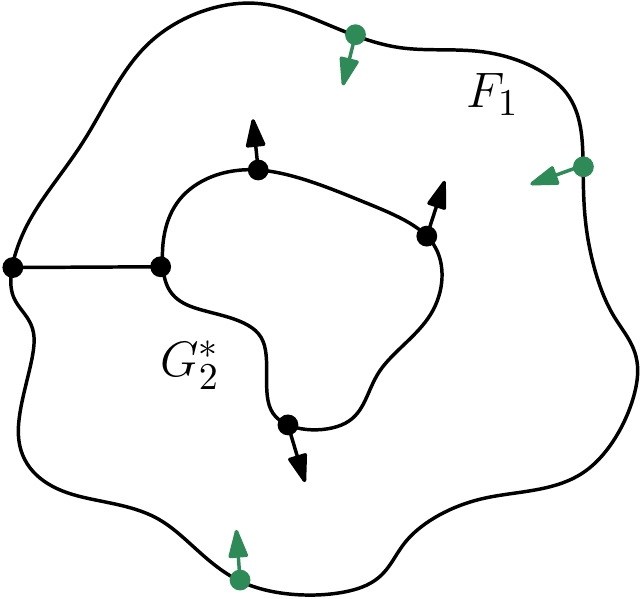}
\\
(a)
&
~~~~
&
(b)
&
~~~~
&
(c)
\end{tabular}
\caption{(a) The realization of $G^*_1$ and $G^*_2$ from Fig.~\ref{fig:edge-cuts}(a), as $P_1$ and $P_2$, (b) $G^*_1$, (c) $G_2^* \cup F_1$.}
\label{fig:decomposition}
\end{figure}
The vertices lying on the face of $F_1$ in $G_1^*$ are ascribed to exactly one of $G_1^*$ or $G_2^*\cup F_1$ as follows. If the vertex is of degree 3 then it goes
to the copy of $F_1$ in $G_1^*$, if it is of degree 2 it ascribes it as indicated by the vertex assignment of $G^*$.  See Fig.~\ref{fig:decomposition}(a), where vertices of $F_1$ in red are
assigned to $G_1^*$ and vertices of $F_1$ in green are assigned to $G_2^*\cup F_1$.

The algorithm can now reconstruct a unique topological graph having $G_2^* \cup F_1$ as a triangulation dual. This process creates a triangle for each vertex in $G_2^*\cup F_1$ whose orientation is unique
for vertices of degree 3 by the location of the edges in $G_2^*\cup F_1$ or given by the hole assignment in $G^*$ for vertices of degree 2. In this case the adjacencies are as prescribed by the edges in $G^*$.

Now for $G_1^*$, if it has no other one-edge cuts, we apply the same process as to $G_2^* \cup F_1$ and obtain a topological embedding. Otherwise we recursively process the one-edge cuts of $G_1^*$ as above and also obtain a topological realization of $G_1^*$.

Once we have the topological representations of each of $P_1$ and $P_2$, we merge them as follows.
First we reinsert all the vertices of the hole face $F_1$, then grow the center of the wheel in $P_1$ into a fat point (Fig.~\ref{fig:reconstruction-triangulation}(a)). Next, replace this fat point with $P_2$ (shown in Fig.~\ref{fig:reconstruction-triangulation}(b)).

\begin{figure}\centering
\begin{tabular}{ccc}
\includegraphics[width=.35\textwidth]{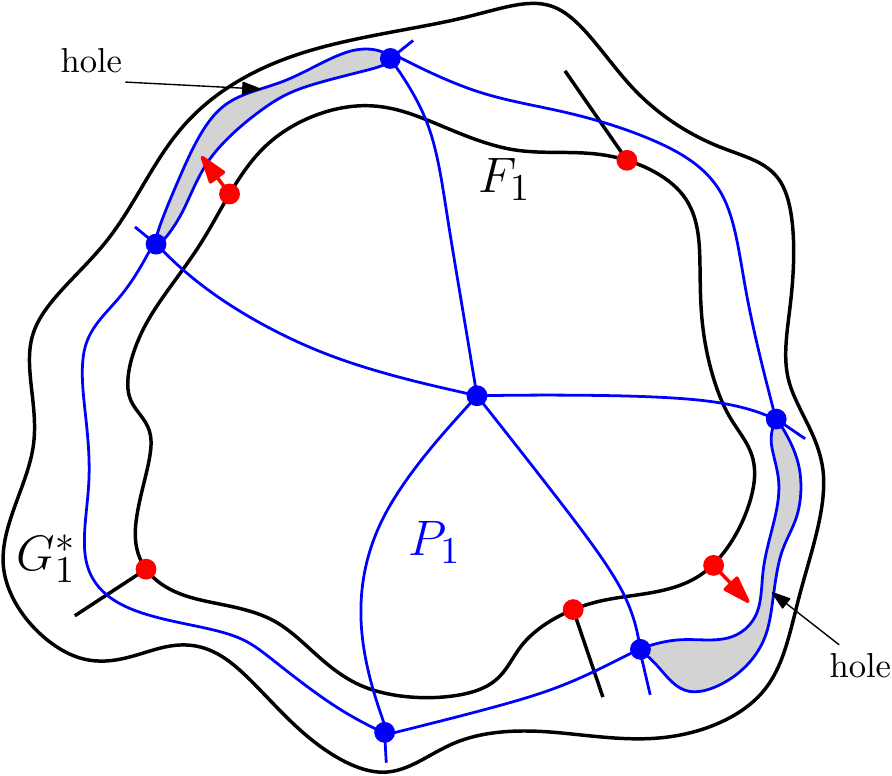}
  &
  ~~~~~~~~~~~
  &
\includegraphics[width=.3\textwidth]{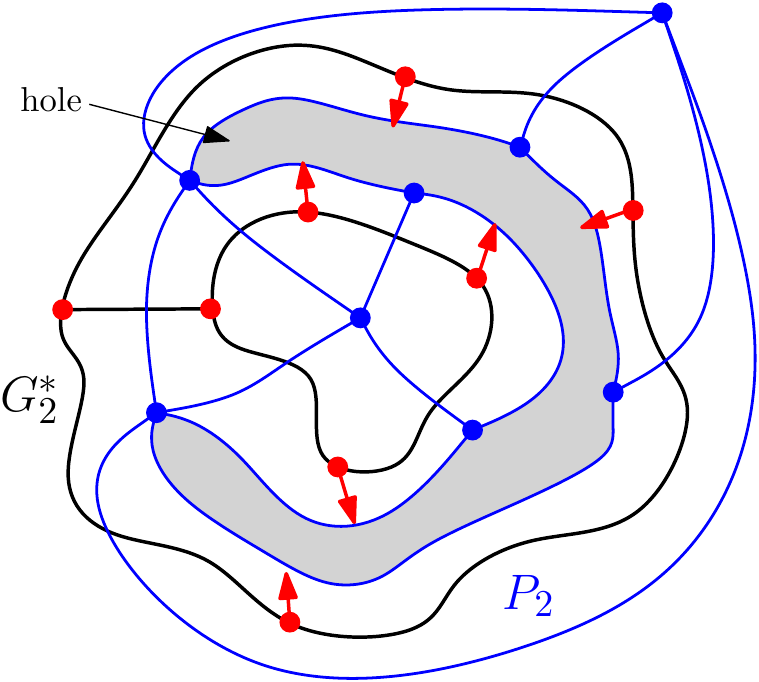}
   \\
   (a)
   &
   ~~~~~~~~~~~
   &
   (b)
\end{tabular}
\caption{ (a) Reconstruction of the triangulation of $P_1$ from $G^*_1$, (b) Reconstruction of the triangulation of $P_2$ from $G^*_2\cup F_1$.}
\label{fig:reconstruction-triangulation}
\end{figure}

If two adjacent vertices in $F_1$ in $G_1^*$ are not adjacent in $G^*$ it means there is an intermediate vertex in $G_2^*\cup F_1$. So the shared edge between the topological triangles in $P_1$ associated to those two vertices in the dual is replaced by a slim triangle corresponding to the dual of the
intermediate vertex in $G^*$ which landed in $G_2^*\cup F_1$ (illustrated by gray-filled triangles in Fig.~\ref{fig:inserted-triangles}).
Thus we can merge the two topological graphs $P_1$ and $P_2$ in a unique way as prescribed by the order of the vertices
in the face $F_1$ in $G^*$. Observe that none of these operations creates crossings, so the merged topological graph is planar. We continue this merging process until we have a topological
representation of the potential triangulated polygon. Again, by construction, this topological graph is planar.

\begin{figure}\centering
\includegraphics[width=.4\textwidth]{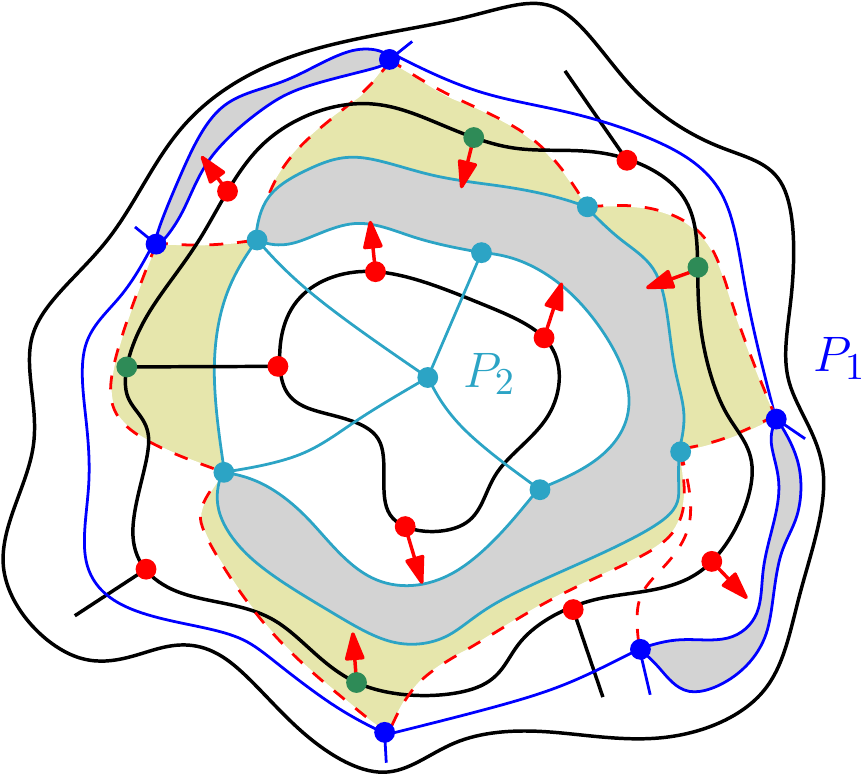}
\caption{Merge of Figures~\ref{fig:reconstruction-triangulation}(a) and \ref{fig:reconstruction-triangulation}(b).}
\label{fig:inserted-triangles}
\end{figure}

We then obtain a straight line embedding of this planar graph using F\'{a}ry's theorem. This is our candidate triangulated polygon plus point at infinity.

To conclude, we need to verify if this is a simple polygon, properly triangulated and with or without interior points as the case may be as follows. If the outer face of the embedding is not a triangle, we reject $G^*$. Otherwise, we consider each of the three vertices in the outer face as a potential point at infinity, i.e., a vertex connected
to all other vertices on the outer face. The remaining structure should additionally be a simple, properly triangulated polygon. If this is the case, we accept $G^*$,
else we reject this point and move to another one in the outer face. If none of the three points satisfy these conditions, we reject $G^*$ as not being the dual of a triangulation of a polygon.
\qed

\begin{figure}[h]\centering
\begin{tabular}{ccc}
  \includegraphics[width=0.35\textwidth]{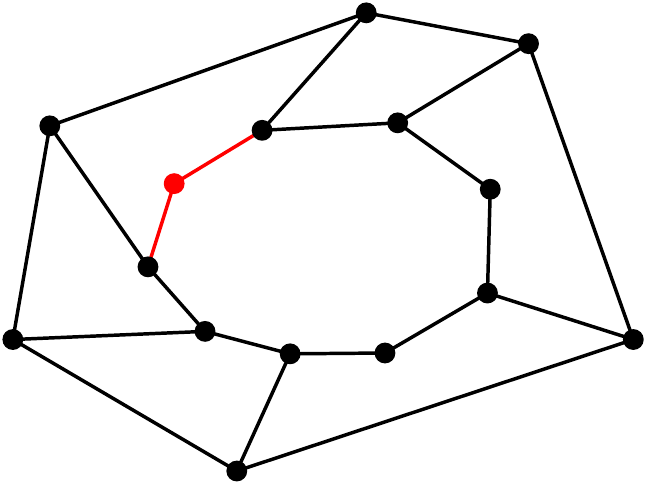}
  &
  ~~~~~~~~~~~~~
  &
   \includegraphics[width=0.35\textwidth]{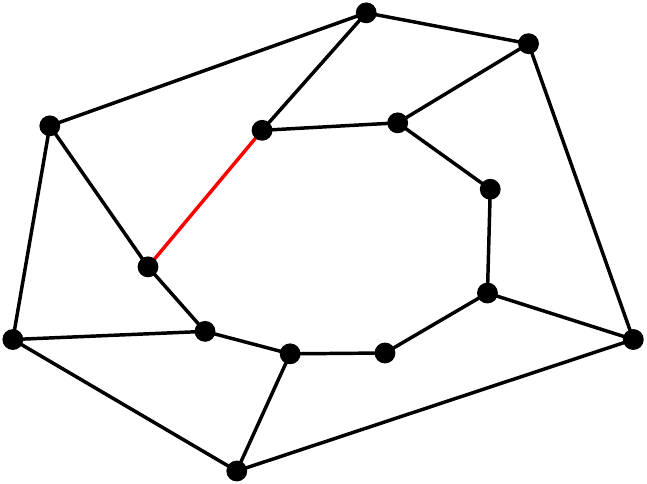}
   \\
   (a)
   &
   ~~~~~~~~~~~~~
   &
   (b)
\end{tabular}
\caption{(a) a dual graph of a triangulation, (b) a graph which does not correspond to a triangulation.}
\label{fig:dual}
\end{figure}

\begin{figure}\centering
   \includegraphics[width=0.4\textwidth]{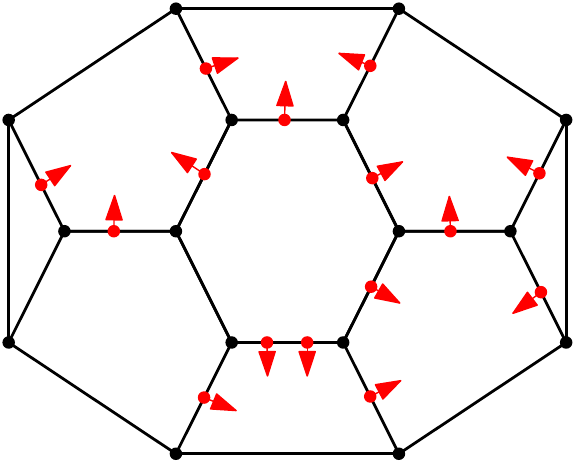}
\caption{A planar graph and an assignment (shown with red arrows) of vertices of degree 2 to faces of the graph.}
\label{fig:assignment}
\end{figure}

\qed
\end{proof}


\paragraph{Geometric TDR- and TDRS-with-known-holes.}\label{subs:gtdrwh}

Given a precise geometric embedding of the input graph, we want to decide if the graph is the triangulation dual of a polygon with holes with or without interior points.

\begin{theorem}\label{thm:geomtdr-with-holes}
The linear program described
in Section~\ref{subs:gtdrwoh} gives a necessary condition for the realization of the geometric TDR- and TDRS-with-known-holes problems with input $\{G^*,\mathcal H\}$  in linear time given the triangulation graph with the circumcenters/centroids of the triangles as vertices.
\end{theorem}
\begin{proof}
Note that if a vertex $v$ is of degree 2 in $G^*$, deciding which face incident to $v$ contains the associated hole can be done by observing the location of the convex angle formed by the two edges of the triangulation perpendicular to the edges incident to $v$ in $G^*$. (For an illustration, see Fig.~\ref{fig:2cases}(a) in which the hole can only reside in the right face.) We then set up an LP as in Theorem~\ref{thm:circumcentres} which gives a potential realization of the triangulation. We then test this solution to verify that the polygon and holes obtained are simple.
\end{proof}

\paragraph{Topological TDR- and TDRS-with-unknown-holes.}
\label{subs:ttdrwuh}

The input for this version of the problem is a planar graph $G^*$ with its face-embedding. However, the total assignment of its vertices of degree 2 to faces of $G^*$ is unknown. Here, we only state that the problem is NP-complete (Theorem~\ref{thm:tdrwh}) and proceed in our analysis. The proof of this claim is provided in Section~\ref{sec:npcomplete}.

\begin{theorem}\label{thm:tdrwh}
Determining if an input graph $G^*$ is the dual of a triangulation
of a polygon with holes and with or without interior points is NP-complete.
\end{theorem}

\paragraph{Combinatorial TDR- and TDRS-with-unknown-holes.}\label{subs:ctdrwh}

In this subsection, the input graph $G^*$ is given by its adjacency matrix.
We will show that the 3-SAT reduction from the topological TDR- and TDRS-with-unknown-holes problems (see Section \ref{sec:npcomplete}) holds as well.
If the embedding found by the combinatorial TDR solver is the same as in the reduction, we would need to solve the 3-SAT problem.
However, it remains to be shown that there does not exist a different embedding with an alternate polygonal realization
and the answer being ``yes'', without this embedding necessarily implying satisfiability of the 3-SAT formula.

Recall that a 3-regular graph has a unique embedding in the plane. We now remove the vertices of degree 2 from
the 3-SAT reduction graph and replace them by an edge, thus giving a 3-regular graph with a unique embedding.
If the combinatorial TDR- and TDRS-with-unknown-holes problems found a different embedding, we can replace the vertices of degree 2
in this alternate embedding with a single edge, thus obtaining a different embedding for the 3-regular graph, which is
a contradiction. Hence, the combinatorial graph obtained from the reduction above has a polygonal realization if and only
if the underlying formula is satisfiable and we obtain:
\begin{theorem}\label{thm:ctdrwh}
The combinatorial TDR- and TDRS-with-unknown-holes problems are NP-complete.
\end{theorem}

\section{NP-Completeness of topological TDR- and TDRS-with-unknown-holes problems}
\label{sec:npcomplete}

In this section, we prove Theorem~\ref{thm:tdrwh}.
Let $X=(x_1,x_2,\ldots,x_m)$ be a set of boolean variables.
Let $\varphi$ be a 3-SAT boolean formula of the type
$\varphi=(a_{11}\vee a_{12}\vee a_{13})\wedge (a_{21}\vee a_{22}\vee a_{23})\wedge \ldots \wedge (a_{n1}\vee a_{n2}\vee a_{n3})$,
where $a_{ij}$ is either $x_k$ or $\neg x_k$ (called a \emph{literal}). We restrict our attention to planar
3-SAT formulae. A~planar 3-SAT formula, by definition, can be represented by a planar graph which has a vertex for every clause and every variable, and has an edge connecting said variable to every clause in which it appears (negated or non-negated).

Planar 3-SAT is known to be NP-complete~\cite{cit:lich}.
We will reduce planar 3-SAT to dual triangulation recognition by
constructing a graph $G^*$ that is the dual of a triangulation of a polygon with holes if and only if given formula $\varphi$ is satisfiable. Our reduction creates $G^*$ which consists of four types of gadgets (Fig.~\ref{fig:4types-faces}(a)--(d) resp.):

\begin{enumerate}
\item variable faces which correspond to variable vertices;
\item clause gadgets which correspond to clause vertices;
\item splitter faces which correspond to some edges connecting a variable vertex to a clause; and
\item absorber gadgets which act as dead ends for extra splitter wires which are not needed.
\end{enumerate}

\begin{figure}[h]\centering
\begin{tabular}{cccc}
   \includegraphics[width=0.2\textwidth]{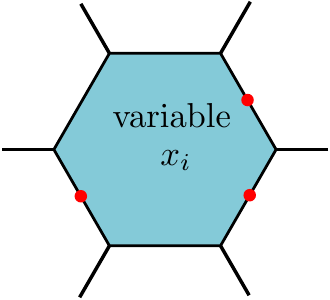}\hspace{0.12cm}
   &
   \includegraphics[width=0.2\textwidth]{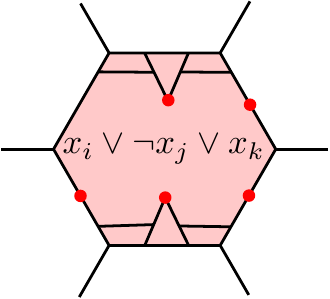}\hspace{0.12cm}
   &
   \includegraphics[width=0.2\textwidth]{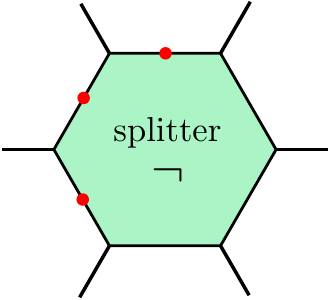}\hspace{0.12cm}
   &
   \includegraphics[width=0.2\textwidth]{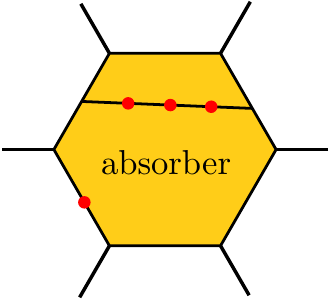}
   \\
   (a)\hspace{0.12cm}
   &
   (b)\hspace{0.12cm}
   &
   (c)\hspace{0.12cm}
   &
   (d)
\end{tabular}
\caption{The types of faces and gadgets of $G^*$: (a) a~variable face,
(b) a clause gadget, (c) a~splitter face and (d) an absorber gadget.} \label{fig:4types-faces}
\end{figure}

Here we describe the purpose of each face and gadget type.
Each face contains some vertices of degree 2 which are only compatible with a triangulation dual if they are adjacent to a hole in the polygon. Each such vertex lies on the boundary between two faces, and there is a choice as to in which of the two faces the hole lies. Our constructions ensures that there cannot be a hole on both faces of a degree 2 vertex. Also, note that even though our examples contain many faces of length 6, the length of faces is in fact determined by given the formula~$\varphi$ and can be arbitrary.

Recall now Lemma \ref{lem:deg2=hole} for the fact that if $G^*$ is the dual graph of a polygon with holes, then each face of $G^*$ that encloses a hole of the polygon, has at least three vertices of degree 2 in said face.

\noindent
{\em Variable face.~~} Each variable face has exactly three vertices of degree 2. This means that either there is a hole in the
variable face (which corresponds to an assignment of {false} to said variable in $\varphi$) or there is no hole and each
of those three vertices of degree 2 have a hole on the other face it belongs to.
This other face is either a part of a clause gadget, a splitter face or an absorber gadget.

\noindent{\em Clause gadget.~~}
The clause gadget (see Fig.~\ref{fig:4types-faces}(b)) has one ``main'' clause face with two vertices of degree 2 shared with other faces with no other vertices of degree 2. This means that there must always be a hole in the clause
face. Each variable contributes a vertex of degree 2 to a clause face (either directly, or via a splitter face). We need
at least three vertices of degree 2 for a hole. Hence, at least one of the degree 2 vertices has the hole in the clause face, otherwise $G^*$ is not the dual graph of a polygon with holes.

If a variable is non-negated in the clause, then the clause face is connected directly to the variable face,
unless we need extra non-negated copies, in which case we use the double-in-series splitter trick (see Splitter face).

If the variable is negated in the clause, the corresponding degree 2 vertex is contributed by a splitter face.

\noindent{\em Splitter face.~~}
The splitter face (see Fig.~\ref{fig:4types-faces}(c)) ``receives'' a degree 2 vertex corresponding to a variable and does two things at the same time: (1) it creates two copies of that variable, and (2) it negates each of them. Hence, the splitter face is always incident to precisely three vertices of degree 2.

The splitter is connected to a variable face or a splitter face by sharing a pair of edges centred around a vertex of degree 2. This is where it ``receives'' the vertex of degree 2 from. It ``passes'' the negated copies of that vertex to another splitter, to an absorber or to a clause gadget again by sharing a pair of edges centred around the copy (a vertex of degree 2).

The splitter face always creates two negated copies of a vertex of degree 2. If only one copy is needed, the other one is passed to a neighbouring absorber face (see below). If a non-negated copy of a vertex of degree 2 is needed, we pass it through another splitter to introduce ``double negation'' (and absorb the redundant copy).

The polygon may or may not contain holes in the splitter faces. If a hole is present, it indicates that the splitter passes a degree 2 vertex forward corresponding to the negated form of the variable.

\noindent{\em Absorber gadget.~~} The absorber gadgets always correspond to parts of a triangulation, regardless of the
rest of the structure of the graph and its polygonal interpretation. Their purpose is to consume unwanted vertices of degree 2 and provide space for holes of a polygon. The vertex of degree 2 is passed to a part of an absorber with three degree 2 vertices, so the face can contain hole regardless whether the degree 2 vertex is assigned to be part of that hole, or not.

\medskip
We construct a graph such that if the variable $x_i$ corresponding to the variable face $F_{x_i}$ is false in a satisfiable assignment of $\varphi$, the degree 2 vertices are assigned to $F_{x_i}$ (the red arrows in our figures point inwards), and if the variable is true in the assignment, then all the degree 2 vertices are assigned to the other face. The construction begins by constructing the planar graph $G_\varphi$, which represents $\varphi$, and embedding it in the plane. Later, we will replace its vertices by corresponding gadgets. However, for this to be possible, the graph needs to be modified first.

Each edge in $G_\varphi$ indicates a ``transfer'' of a degree 2 vertex. We first need to modify the graph so that the vertices representing variables of $\varphi$ have degree precisely~3. If the degree of such a vertex $x_i$ is less than 3, we increase it by attaching the required number of new vertices (those will be replaced by absorber gadgets). If the degree of $x_i$ is more than 3, we reduce its degree by detaching $deg(x_i)-2$ edges consecutive in cyclic order around $x_i$ (with respect to the embedding of $G_\varphi$), routing them into a new splitter vertex $s$, and connecting $x_i$ to the splitter. Note that this negates the variable $x_i$, so some of the edges may need to be routed through another splitter to cancel this negation. This produces a plane graph where $x_i$ has degree $3$ and the splitter vertex $s$ has degree $deg(x_i) - 1$. Repeatedly applying this construction, the degree of $s$ can be decreased to $3$.

By the construction above, we obtain a plane graph $H_\varphi$ where all the variable, splitter and clause vertices have degree $3$, and absorbers have degree $1$. Now we replace every vertex with the respective gadget so that every edge in $H_\varphi$ is represented by a degree 2 vertex surrounded by edges shared between two gadgets, and so that the topology of the gadgets is equivalent to the embedding of $H_\varphi$ (this is similar to constructing a dual graph of $H_\varphi$). Let us denote the obtained graph by $H^*$. The embedding of $H^*$ contains some ``void'' areas between some gadgets. Those areas can be suitably attributed to faces of gadgets (by removing edges). We obtain graph $G^*$, call it the \emph{gadget graph of $\varphi$}, formed by vertices of degree $3$ and $2$.
See Fig.~\ref{fig:example-graph-pl3SAT}(b) for an example of a formula and the corresponding gadget graph.


We can now argue that graph $G^*$ is a triangulation dual if and only if the formula~$\varphi$ is satisfiable.

\begin{lemma}
\label{lem:equivalence}
The gadget graph $G^*$ of formula $\varphi$ is dual of a triangulation of a simple polygon with holes if and only if $\varphi$ is satisfiable.
\end{lemma}
\begin{proof}
If the formula is satisfiable this means that there is a unique true or false assignment for each variable.
We use this assignment to decide if the degree 2 vertices in a variable face have a corresponding hole in this face or outside of it. Either choice forces the assignment in the other face. In an absorber face,
that face will have a hole regardless of the variable assignment. In a splitter face,
the hole (no-hole) choice in the variable face forces a no-hole (hole, resp.) choice
in the splitter face since there is one vertex of degree 2 now out (correspondingly now
in) of the splitter face.
In the case of the clause face, if the vertex of degree 2 corresponding to a variable
is true and it appears non-negated (or the variable is false and it appears negated)
then it has no hole on the other side. This means
the clause face has now at least three vertices of degree 2 with hole inside the clause face,
and we can now safely place a hole in the clause face.


Observe that since the formula is satisfiable, every clause face has at least one literal
which is true and hence a third vertex of degree 2. So, variable and splitter
faces have consistent holes by our choice of their placement; absorber faces are indifferent to
our choice of hole locations; and clause faces always have consistent holes since at least
one of the vertices of degree 2 has no hole on the other side.

Now assume that the graph is realizable as a triangulation dual. Then assign to each variable
in $\varphi$ false if there is a hole in the corresponding variable face and true if there
is no hole in the variable face. Observe that the parity of splitter faces connecting the
variable face and the clause face correspond by construction to whether the variable
appears negated or non-negated in the clause. Thus it follows that if there is (or resp. there is not)
a hole in the variable face, then the associated vertex of degree 2 is assigned to the clause
face if and only if the variable appears non-negated (negated, resp.). Since the clause
face was realizable as the dual of a triangulation with a hole, it follows that at least
one of the literals appearing in the clause is set to true and hence, the clause is
satisfied.
\qed
\end{proof}

\medskip
Fig.~\ref{fig:example-graph-pl3SAT}(b) illustrates a graph $G^*$ associated to the planar 3-SAT formula
 $\varphi = (x_1\vee\neg x_2\vee x_3)\wedge (x_2\vee x_3\vee x_4)\wedge (x_1\vee\neg x_3\vee\neg x_4)$
together with a correct assignment of vertices of degree 2 to faces of $G^*$. Thus $G^*$ is a dual graph of a triangulation, which implies
that the formula is satisfiable for the truth assignment $U$: 
$(x_1,x_2,x_3,x_4)^U = (T,F,T,F)$.

\begin{figure}[h]\centering
\begin{tabular}{cc}
   \includegraphics[width=0.4\textwidth]{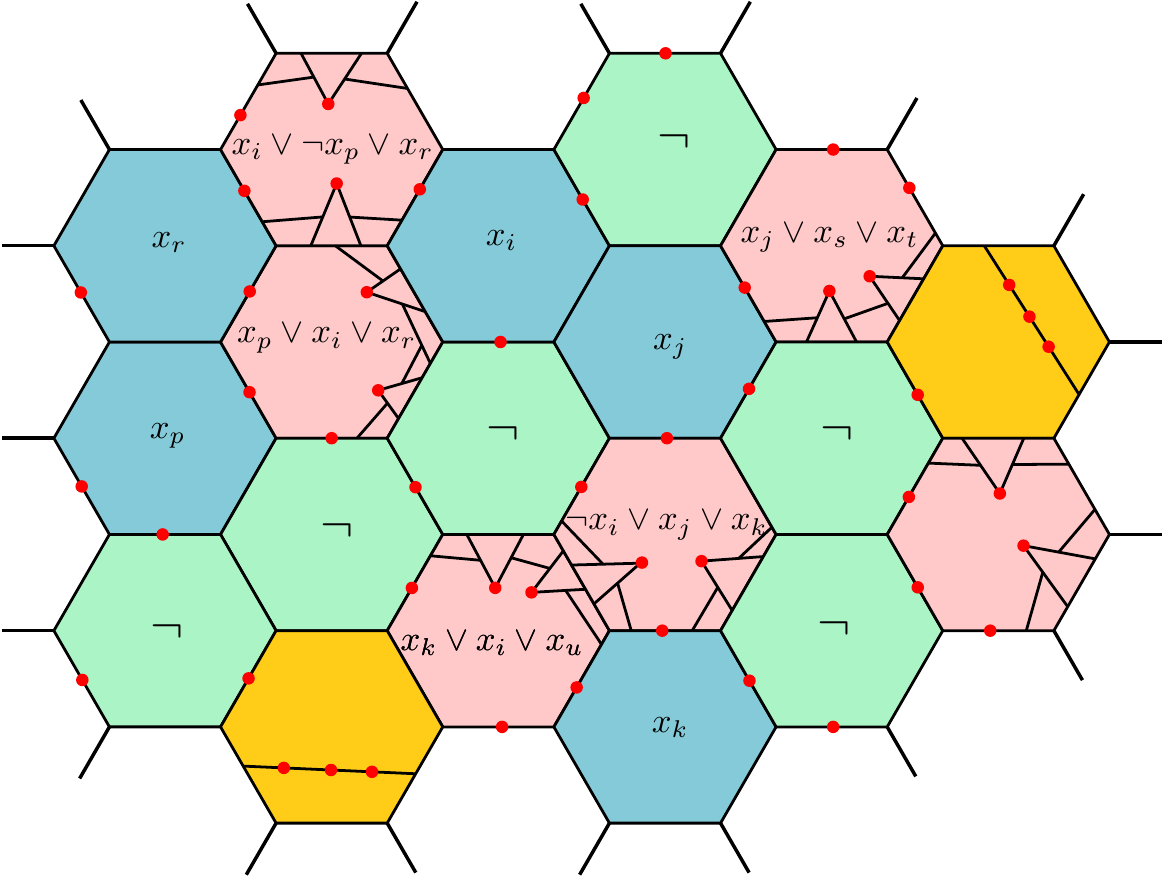}
   &
   \includegraphics[width=0.4\textwidth]{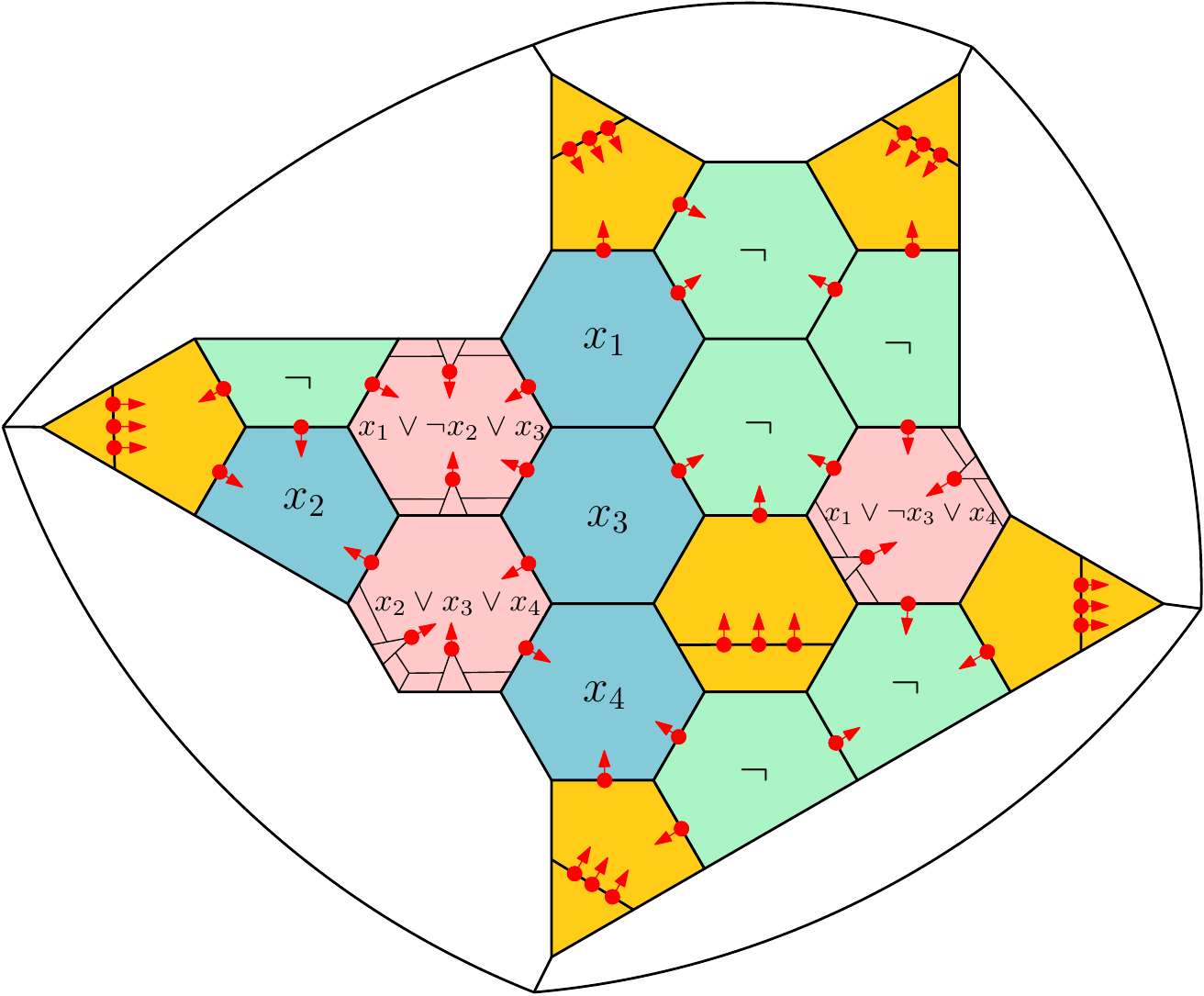}\\
   (a)
   &
   (b)
   \end{tabular}
   \caption{(a) A part of the gadget graph associated to a planar 3-SAT formula,
   (b) An example of a gadget graph $G^*$ corresponding to the planar 3-SAT formula
   $\varphi = (x_1\vee\neg x_2\vee x_3)\wedge (x_2\vee x_3\vee x_4)\wedge (x_1\vee\neg x_3\vee x_4)$.
   The red arrows show the assignment of vertices of degree $2$ to faces of $G$.}
\label{fig:part-planar3SAT}\label{fig:example-graph-pl3SAT}
\end{figure}

\section{Conclusions and Open Questions}
\label{sec:conclusion}

We provided an exhaustive analysis of the triangulation dual recognition problem. We showed that some of them can be solved in linear time and some of them are NP-complete. Our work focused on duals of general triangulations of simple polygons.
We proposed several models for the geometric setting.
We presented a method which in linear time finds a candidate solution, or rejects. The candidate solution needs to be further tested. 
As our approach is not capable of enumerating all the candidate solutions, it remains an open problem if
recognition is possible under either of these models. Any bounds for necessary and/or sufficient conditions under other choices for triangle representatives are open too.

\end{document}